\documentclass[11pt]{article}
\usepackage{times}
\usepackage{amsmath,amsfonts,amssymb,amsthm}
\usepackage{enumerate}
\usepackage{authblk}
\usepackage[letterpaper, margin=1in]{geometry}
\usepackage{tikz}
\usepackage{etoolbox}

\newtoggle{conf}
\ifx\confver\undefined
	\togglefalse{conf}
\else
	\toggletrue{conf}
\fi

\begin{document}

\newtheorem{theorem}{Theorem}
\newtheorem{innercustomthm}{Theorem}
\newenvironment{customthm}[1]
  {\renewcommand\theinnercustomthm{#1}\innercustomthm}
  {\endinnercustomthm}
\newtheorem{proposition}{Proposition}
\newtheorem{corollary}{Corollary}
\newtheorem{innercustomcor}{Corollary}
\newenvironment{customcor}[1]
  {\renewcommand\theinnercustomcor{#1}\innercustomcor}
  {\endinnercustomcor}
\newtheorem{lemma}{Lemma}
\newtheorem{claim}{Claim}
\newtheorem{definition}{Definition}
\newtheorem*{remark}{Remark}
\newcommand{\diu}{dynamic interval union}
\newcommand{\bps}{batch partial sum}
\newcommand{\es}{evenly-spreading}
\newcommand{\mi}{multi-index}
\newcommand{\mm}{model $\mathcal{M}$}
\newcommand{\E}{\ensuremath{\mathop{\mathbb{E}}}}
\newcommand{\rev}{\ensuremath{\mathrm{rev}}}

\title{Cell-probe Lower Bounds for Dynamic Problems via a New Communication Model}
\author{Huacheng Yu}
\affil{Stanford University}
\date{}

\maketitle

\begin{abstract}
In this paper, we develop a new communication model to prove a data structure lower bound for the \diu{} problem. The problem is to maintain a multiset of intervals $\mathcal{I}$ over $[0, n]$ with integer coordinates, supporting the following operations:
\begin{itemize}
	\item
		\verb+insert(a, b)+: add an interval $[a, b]$ to $\mathcal{I}$, provided that $a$ and $b$ are integers in $[0, n]$;
	\item
		\verb+delete(a, b)+: delete a (previously inserted) interval $[a, b]$ from $\mathcal{I}$;
	\item
		\verb+query()+: return the total length of the union of all intervals in $\mathcal{I}$.
\end{itemize}

It is related to the two-dimensional case of Klee's measure problem. We prove that there is a distribution over sequences of operations with $O(n)$ insertions and deletions, and $O(n^{0.01})$ queries, for which any data structure with any constant error probability requires $\Omega(n\log n)$ time in expectation. Interestingly, we use the sparse set disjointness protocol of H\aa{}stad and Wigderson [ToC'07] to speed up a reduction from a new kind of nondeterministic communication games, for which we prove lower bounds.

For applications, we prove lower bounds for several dynamic graph problems by reducing them from \diu{}.
\end{abstract}

\newpage

\section{Introduction}
In computational geometry, Klee's measure problem~\cite{Klee77, Bent77, Chan13} is the following: given $N$ rectangular ranges (axis-parallel hyperrectangles) in $d$-dimensional space, compute the volume of their union.

A classic sweep-line algorithm by Bentley~\cite{Bent77} solves the $d=2$ case in $O(N\log N)$ time: consider the line $x=x_0$ with $x_0$ continuously increasing from $-\infty$ to $\infty$; the length of the intersection of this line and the union may change only when it reaches the left or right border of a rectangle. Bentley's algorithm uses a segment tree to dynamically maintain the length of the intersection efficiently. Surprisingly, this is the best known algorithm even for an intriguing special case: all coordinates are integers within a polynomially bounded range $[0, n]$. In this case, the segment tree in Bentley's algorithm is essentially used to solve the following dynamic problem, which we call the \emph{\diu{} problem}:

Maintain a multiset $\mathcal{I}$ of intervals with integer coordinates in $[0, n]$, supporting the following operations:

\begin{itemize}
	\item
		\verb+insert(a, b)+: add an interval $[a, b]$ to $\mathcal{I}$, provided that $a$ and $b$ are integers in $[0, n]$;
	\item
		\verb+delete(a, b)+: delete a (previously inserted) interval $[a, b]$ from $\mathcal{I}$;
	\item
		\verb+query()+: return the total length of the union of all intervals in $\mathcal{I}$.
\end{itemize}

The segment tree data structure solves this dynamic problem with $O(\log n)$ insertion and deletion time, and $O(1)$ query time. For the application to 2D Klee's measure problem, there are $N$ insertions, $N$ deletions and $N$ queries to the data structure. A natural question to ask here is whether we can improve the insertion and deletion time. However, there is a very simple reduction from the partial sum problem showing that the slowest operation among insertion, deletion and query needs to take $\Omega(\log n)$ time (see Appendix~\ref{sectredps}). Moreover, P\v{a}tra\c{s}cu and Demaine~\cite{PD04b, PD06} showed an optimal trade-off between update and query time for partial sum, which can be carried over via the reduction to show that if both insertion and deletion need to be done in $O(\epsilon \log n)$ time, then query has to take $\Omega(2^{1/\epsilon}\log n)$ time. 

This seems to be the end of the story. However, in fact, there is no $o(N\log N)$ time algorithm known even for $n=N^{0.51}$.\footnote{There is a simple linear time algorithm for the $n\leq \sqrt{N}$ case.} When we apply the above dynamic problem to this case, there will be $N$ insertions, $N$ deletions, and only $n=N^{0.51}$ queries! There will be far fewer queries than insertions and deletions. The argument above does not rule out the possibility of having a \diu{} data structure with $o(\log n)$ update time, and $O(n^{0.9})$ query time. It is even possible to have a data structure with $O(1)$ insertion and deletion time, and $O(n^{0.1})$ query time. Having such a data structure would give a linear time algorithm for the above special case of Klee's measure problem, making a breakthrough on this 40-year-old problem. 

Unfortunately, we show that such data structure does not exist, even if we allow randomization, amortization and constant error probability. 

\begin{theorem}\label{thmdiu}
	For any $\epsilon>0$ and integer $n\geq 1$, there is a distribution over operation sequences to the \diu{} problem over $[0, n]$, with $O(n^{1-\epsilon})$ insertions and deletions, and $O(n^{\epsilon})$ queries, for which any data structure that correctly answers all queries simultaneously with probability $\geq 95\%$ must spend $\Omega(\epsilon^2 n^{1-\epsilon}\log n)$ probes in expectation, in the cell-probe model with word size $\Theta(\log n)$. 
\end{theorem}
We define the cell-probe model in \iftoggle{conf}{Appendix~\ref{sectcp_app}}{Section~\ref{subsectcp}}. 

It is an easy exercise to show that we can use the hard distribution from the theorem to obtain a new hard distribution with $\Theta(n^{1/0.51})=\Theta(N)$ insertions and deletions and $n$ queries, such that any data structure requires $\Omega(N\log N)$ time on it. This lower bound rules out the possibility of using the plain-vanilla sweep-line algorithm with a sophisticated data structure to solve 2D Klee's measure problem faster than the classic algorithm. As a corollary, the theorem also implies that $o(\log n)$ insertion and deletion time leads to an almost linear lower bound on query time. 

The type of running time considered by Theorem~\ref{thmdiu} is more general than the amortized time. Having amortized time $t(n)$ usually means that the first $k$ operations take at most $k\cdot t(n)$ time for every $k$, but here, even if we fix the number of operations in advance, and the data structure is allowed to use heavy preprocessing in order to optimize the total running time, there is still no way to break the lower bound. This notion of running time is usually what we care about, when applying a data structure to solve some computational problem. The only catch is that the data structure is online: it must output an answer before seeing the next operation, which makes it different from an offline computational problem. Moreover, we claim without proof the following theorem that using the same hard distribution, the same lower bound holds for data structures that correctly answer any constant fraction of the queries in expectation. 

\begin{customthm}{1$'$}
	For any $\epsilon>0$ and integer $n\geq 1$, there is a distribution over operation sequences to the \diu{} problem over $[0, n]$, with $O(n^{1-\epsilon})$ insertions and deletions, and $O(n^{\epsilon})$ queries, for which any data structure with expected fraction of correct answers at least $\delta$ must spend $\Omega_{\delta}(\epsilon^2 n^{1-\epsilon}\log n)$ probes in expectation for any $\delta\in (0, 1]$, in the cell-probe model with word size $\Theta(\log n)$. 
\end{customthm}

We prove Theorem~\ref{thmdiu} via a reduction from a more accessible intermediate problem called \emph{\bps{}}, for which we prove a lower bound directly. The \bps{} problem asks to maintain $K$ sequences $(A_{i,j})_{i\in [K],j\in [B]}$ of length $B$ over a finite field $\mathbb{F}_p$, supporting the following operations to the sequences:\footnote{In this paper, $[K]$ stands for the set of positive integers $\{1,2,\ldots,K\}$.}

\begin{itemize}
	\item
		\verb+update(+$\mathbf{j}$\verb+, +$\mathbf{v}$\verb+)+: for all $i\in [K]$, set $A_{i,j_i}$ to value $v_i$;
	\item
		\verb+query(+$\mathbf{j}$\verb+)+: return $\sum_{i\in [K]}\sum_{l\leq j_i} A_{i,l}$,
\end{itemize}
provided that $\mathbf{j}$ and $\mathbf{v}$ are vectors of length $K$, and $j_i\in [B]$, $v_i\in \mathbb{F}_p$. Basically, we need to maintain $K$ independent copies of the partial sum problem, except that when answering queries, instead of returning $K$ individual prefix sums, we only need to return the sum of these $K$ numbers. 

\begin{theorem}\label{thmbps}
For large enough integers $K,B,p$ with $p\geq B$, there is a distribution over operation sequences to the \bps{} problem with $O(B)$ updates and $O(B)$ queries, for which any data structure that correctly answers all queries simultaneously with probability $\geq 95\%$ must spend $\Omega(KB\log^2 B/(w+\log B))$ probes in expectation, in the cell-probe model with word size $w$.
\end{theorem}

Note that when $K,B,p$ are all polynomials in $n$ and $w=\Theta(\log n)$ for some $n$, the lower bound becomes $\Omega(K\log n)$ per operation. Therefore, in this case, the best thing to do is just to use $K$ partial sum data structures to maintain the $K$ sequences independently. However, the techniques we use in the proof are very different from the proof of the lower bound for partial sum by P\v{a}tra\c{s}cu and Demaine~\cite{PD04a}. See Section~\ref{subsecttc} for an overview. 

Moreover, we also apply our main theorem to prove lower bounds for three dynamic graph problems: dynamic \#SCC, dynamic weighted s-t shortest path, and dynamic planar s-t min-cost flow (see Appendix~\ref{sectcata} for formal definitions). We prove that for these problems, under constant error probability, if we spend $o(\log n)$ update time, then queries must take $n^{1-o(1)}$ time. Note that a previous result by P\v{a}tra\c{s}cu and Thorup~\cite{PT11} also implies the same trade-off for the first two problems under zero error. 

\begin{customcor}{\iftoggle{conf}{2}{\ref{cordyngraph}}}
For the following three dynamic graph problems:
\begin{enumerate}[(a)]
	\item
		dynamic \#SCC,
	\item
		dynamic planar s-t min-cost flow,
	\item
		dynamic weight s-t shortest path,
\end{enumerate}
any data structure with amortized expected update time $o(\log n)$, and error probability $\leq 5\%$ under polynomially many operations must have amortized expected query time $n^{1-o(1)}$. 
\end{customcor}

\def\cellprobepre{The cell-probe model of Yao~\cite{Yao81} is a strong non-uniform computational model for data structures. A data structure in the cell-probe model has access to a set of memory cells. Each cell can store $w$ bits. The set of cells is indexed by $w$-bit integers, i.e., the address is in $[2^w]$. $w$ is usually set to be $\Omega(\log n)$, where $n$ is the amount of information the data structure needs to handle. \footnote{$n$ is usually a polynomial in the number of operations. }

To access the memory, the data structure can \emph{probe} a cell, which means that it can look at the content of the cell, and then optionally overwrite it with a new value. During an operation, the data structure based on the parameters of the operation decides which cell to probe the first, then based on the parameters and the information from the first probe, decides the cell to probe next, etc. Each cell-probe (including both the address and the new value) the data structure performs may be an arbitrary function of the parameters of the operation and the contents in the cells previously probed during this operation. If the operation is a query, in the end, the data structure returns an answer based on the parameters and the contents in all cells probed during this operation. The update (query resp.) time is defined to be the number of cells probed during a(n) update (query resp.) operation. 

The cell-probe model only counts the number of memory accesses during each operation, making itself a strong model, e.g., it subsumes the word-RAM model. Thus, data structure lower bounds proved in this model will hold in various other settings as well. }
\iftoggle{conf}{}{
\subsection{Cell-probe Model}\label{subsectcp}
\cellprobepre
}

\subsection{Previous Cell-probe Lower Bounds}\label{subsectplb}

In 1989, Fredman and Saks introduced the \emph{chronogram} method to prove an $\Omega(\log n/\log\log n)$ lower bound for partial sum in their seminal paper~\cite{FS89}. The lower bound is tight for maintaining a sequence of $\{0,1\}$s. $\Omega(\log n/\log\log n)$ was also the highest lower bound proved for any explicit data structure problem for a long time.

In 2004, P\v{a}tra\c{s}cu and Demaine~\cite{PD04a} broke this $\log n/\log\log n$ barrier using a new approach: the \emph{information transfer tree} technique. They proved an $\Omega(\log n)$ lower bound for the partial sum problem with numbers in $[n]$. Moreover, using this new technique, one can prove an update-query time trade-off of $t_q\log\frac{t_u}{t_q}=\Omega(\log n)$, where $t_u$ is the update time and $t_q$ is the query time, while earlier approaches can only prove $t_q\log t_u=\Omega(\log n)$. Later on, the information transfer tree technique has been used to prove several other data structure lower bounds~\cite{PD04b, CJS15, CJ11, CJS13}.

In 2012, there was a breakthrough by Larsen~\cite{Larsen12a} on dynamic data structure lower bounds. Larsen combined the chronogram method with the \emph{cell sampling} technique of Panigraphy, Talwar and Wieder~\cite{PTW10}, and proved an $\Omega((\log n/\log\log n)^2)$ lower bound for the 2D orthogonal range counting problem. This lower bound is also the highest lower bound proved for any explicit dynamic data structure problem so far. Similar approaches were also applied later~\cite{Larsen12b, CGL15}. 

All above techniques can only be used to prove a relatively smooth trade-off between update and query time. However, P\v{a}tra\c{s}cu and Thorup~\cite{PT11} used a new idea to prove a sharp trade-off for the dynamic connectivity problem in undirected graphs. They proved that if one insists on $o(\log n)$ insertion and deletion time, query has to take $n^{1-o(1)}$ time. Besides the sharp trade-off, they also introduced the \emph{simulation by communication games} of the data structure. They first decomposed the entire execution of the data structure on a sequence of operations into several communication games. For each communication game, they showed how to turn a ``fast'' data structure into an efficient communication protocol. Then they proved a communication lower bound for each game. Summing all these lower bounds up establishes a lower bound on the total number of probes in the entire execution.

\subsection{Technical Contributions}\label{subsecttc}
Although the \bps{} problem looks similar to the partial sum problem, it seems hopeless to apply the information transfer technique directly to solve our problem. It is due to a critical difference between the two problems: in partial sum, the lower bound proved roughly equals to the number of bits in the answer to a query; while in batch partial sum, the lower bound we aim at is much larger than the size of an answer. The proof in \cite{PD04a} heavily relies on the fact that in partial sum problem, after fixing the values in a lot of entries in the sequence, as long as there is still one summand in the prefix sum left uniformly at random, the sum will also be uniformly at random. Therefore, we will need to learn a certain amount of information from the memory to figure out the answer. If we apply the same technique in \bps{}, as an answer still contains only $\log n$ bits of information, we will again get a lower bound of $\Omega(\log n)$ per operation, while we aim at $\Omega(K\log n)$. The cell sampling technique has a similar issue. It can only be applied when the number of bits used to describe a query is comparable with the number of bits used in an answer. 

The main idea of our proof is to use the simulation by communication games technique mentioned in Section~\ref{subsectplb}. After decomposing into communication games, there are two things to prove: a ``fast'' data structure implies an efficient communication protocol, and no efficient communication protocol exists. The choice of communication model for the game is crucial. If we use a too weak communication model, it would be hard to take advantage of the model to design an efficient protocol given fast data structure. If the communication model we use is too strong, it would be difficult or even impossible to prove a communication lower bound, especially when small chance of error is allowed. P\v{a}tra\c{s}cu and Thorup gave two different simulations: one in the deterministic setting, the other in the nondeterministic setting. The deterministic simulation itself (transforming a data structure into a communication protocol) is not efficient enough to achieve our lower bound. The nondeterministic model, in our case of allowing error, would correspond to the distributional $\mathrm{MA^{cc}}\cap \textrm{co-}\mathrm{MA^{cc}}$ model. It is particularly difficult to prove a lower bound in this model. In our application, the communication problem we want to prove a lower bound for is closely related to the inner product problem. Namely, Alice and Bob get $n$-dimensional binary vectors $x$ and $y$ respectively and the goal is to compute the inner product $\left<x, y\right>$ over $\mathbb{F}_2$. There is a clever $\mathrm{MA^{cc}}\cap \textrm{co-}\mathrm{MA^{cc}}$ protocol by Aaronson and Wigderson~\cite{AW08} which solves the inner product problem with only $\tilde{O}(\sqrt{n})$ bits of communication. It has much less cost than expected, which suggests that it might even be impossible to prove a desired lower bound for our problem in this strong model. 

To overcome this obstacle, we define a new communication model (see Section~\ref{sectdiubps}), which is weaker than $\mathrm{MA^{cc}}\cap \textrm{co-}\mathrm{MA^{cc}}$, so that we are capable of proving a desired communication lower bound. Moreover, we will be able to achieve the same performance of transforming data structure into protocol as in the nondeterministic model. Interestingly, in order to have less requirement on the power of communication model, we use an elegant protocol for computing sparse set disjointness by H\aa{}stad and Wigderson~\cite{HW07} as a subroutine:

\begin{theorem}[H\aa{}stad and Wigderson]\label{thmsparse}
In the model of common randomness, $R_0(\mathrm{DISJ}^n_k) = O(k)$ for instances of disjoint sets and $R_0(\mathrm{DISJ}^n_k) = O(k+\log n)$ for non-disjoint sets.
\end{theorem}

$\mathrm{DISJ}^n_k$ is the following problem: Alice and Bob get sets $X$ and $Y$ of size $k$ over a universe $[n]$ respectively, their goal is to compute whether the two sets are disjoint. $R_0(\mathrm{DISJ}^n_k)$ stands for the minimum expected communication cost by any zero-error protocol which computes $\mathrm{DISJ}^n_k$. 

As we will see later, this new communication model has the power of nondeterminism. Also it is restricted enough so that we can apply the classic techniques for proving randomized communication lower bounds. Using this model, we prove the first sharp update-query trade-off under constant probability of error. 

\subsection{Overview}
The remainder of this paper is organized as follows. In Section~\ref{sectdiubps}, we present the reduction from \bps{} to \diu{}, and define the new communication model and the new simulation. In Section~\ref{sectcomm}, we prove a communication lower bound in this new model, which completes the proof of our main result. \iftoggle{conf}{Some proof details and the applications of our main theorem to dynamic graph problems can be found in the full version. }{In Section~\ref{sectapp}, we apply the main theorem to several dynamic graph problems. }Finally, we conclude with some remarks in Section~\ref{sectremark}.

\def\sectprecomm{
In the classic deterministic communication complexity setting~\cite{Yao79}, two players Alice and Bob receive inputs $x\in \mathcal{X}$ and $y\in \mathcal{Y}$ respectively. Their goal is to collaboratively evaluate a function $f$ on their joint input $(x, y)$. The players send bits to each other according to some predefined protocol. At each step, the protocol must specify which player sends the next bit based on the transcript (the bits sent so far). The sender decides to send a bit 0 or 1 based on the transcript and his/her input. It the end, the answer $f(x,y)$ can only depend on the entire transcript. In the setting with public randomness, the players have access to a common random binary string of infinite length. Besides the sender's input and the transcript, each message may also depend on these random bits. The players have infinite computational power. The cost of a protocol is the number of bits communicated, i.e., the length of the transcript.

\begin{definition}
	For function $f$ with domain $\mathcal{X}\times \mathcal{Y}$, the matrix $M(f)$ is  a $|\mathcal{X}|\times |\mathcal{Y}|$ matrix, with rows indexed by $\mathcal{X}$ and columns indexed by $\mathcal{Y}$. The entry in row $x$ and column $y$ is the function value $f(x, y)$.
\end{definition}

\begin{definition}
A \emph{combinatorial rectangle} or simply a \emph{rectangle} in $M(f)$ is a set $X\times Y$ for $X\subseteq \mathcal{X}$ and $Y\subseteq \mathcal{Y}$. 
\end{definition}

\begin{definition}
A \emph{monochromatic rectangle} in $M(f)$ is a combinatorial rectangle in which the function value does not vary. 
\end{definition}

\begin{definition}
Let $\mu$ be a distribution over $\mathcal{X}\times \mathcal{Y}$. A $\alpha$-\emph{monochromatic rectangle} under $\mu$ is a combinatorial rectangle $X\times Y$ such that there is a function value $v$, $\alpha\mu(X\times Y)\leq \mu((X\times Y)\cap f^{-1}(v))$, i.e., a combinatorial rectangle with at least $\alpha$-fraction of the input pairs having the same function value. 
\end{definition}

A classic result~\cite{BookKN97} in communication complexity is that every protocol in the deterministic setting with worst-case communication cost $C$ induces a partitioning of $M(f)$ into $2^C$ monochromatic rectangles. Each rectangle corresponds to one possible transcript, i.e., when the players are given an input pair in this rectangle, the corresponding transcript will be transmitted. A similar result shows that every randomized protocol with low error probability induces a partitioning into rectangles, such that \emph{most} of the rectangles are \emph{nearly monochromatic} ($\alpha$-monochromatic with $\alpha$ close to $1$). Proving there is no large monochromatic rectangle or nearly monochromatic rectangle in $M(f)$ would imply communication lower bounds in deterministic or randomized setting respectively.
}

\iftoggle{conf}{}{
\section{Preliminaries on Communication Complexity}\label{sectpre}
\sectprecomm
}

\section{Lower Bounds for Dynamic Interval Union and Batch Partial Sum}\label{sectdiubps}

In this session, we will prove our main result, a lower bound for the \diu{} problem, via a reduction from the \bps{} problem. 

\begin{customthm}{\ref{thmdiu}}
For any $\epsilon>0$ and integer $n\geq 1$, there is a distribution over operation sequences to the \diu{} problem over $[0, n]$, with $O(n^{1-\epsilon})$ insertions and deletions, and $O(n^{\epsilon})$ queries, for which any data structure that correctly answers all queries simultaneously with probability $\geq 95\%$ must spend $\Omega(\epsilon^2 n^{1-\epsilon}\log n)$ probes in expectation, in the cell-probe model with word size $\Theta(\log n)$. 
\end{customthm}

\begin{customthm}{\ref{thmbps}}
For large enough integers $K,B,p$ with $p\geq B$, there is a distribution over operation sequences to the \bps{} problem with $O(B)$ updates and $O(B)$ queries, for which any data structure that correctly answers all queries simultaneously with probability $\geq 95\%$ must spend $\Omega(KB\log^2 B/(w+\log B))$ probes in expectation, in the cell-probe model with word size $w$. 
\end{customthm}

The idea of this reduction is similar to the proof of Proposition~\ref{redps} in Appendix~\ref{sectredps}. 

\begin{proof}[Proof of Theorem~\ref{thmdiu}]

Take prime $p=\Theta(n^{\epsilon})$, $B=\Theta(n^{\epsilon})$ with $B\leq p$, and $K=n/Bp=\Theta(n^{1-2\epsilon})$. We are going to show that we can solve the \bps{} problem with this setting of the parameters given a \diu{} data structure over $[0, n]$. We first concatenate the $K$ sequences into one long sequence of length $KB$, such that $A_{i,j}$ will be $((i-1)\cdot B+j)$-th number in the long sequence, and try to maintain the whole sequence using one \diu{} data structure. Then we associate each number in the long sequence with a segment of length $p$ in $[0, n]$ such that the $k$-th number in the long sequence is associated with $[(k-1)\cdot p,k\cdot p]$. We use the length of interval in the associated segment to indicate the value of the number. That is, we always maintain the invariant that for $k$-th number in the long sequence with non-zero value $v$, we have exactly one interval $[(k-1)\cdot p, (k-1)\cdot p+v]$ intersecting its associated segment. 

To set $k$-th number to a new value $v'$, if before the operation it had value $v\neq 0$, we first call $\texttt{delete}((k-1)p, (k-1)p+v)$ to reset the number. Then if $v'\neq 0$, we call $\texttt{insert}((k-1)p, (k-1)p+v')$ to update its new value to $v'$. Therefore, as an \verb+update+ of the \bps{} problem is just setting $K$ numbers to new values, it can be implemented using $O(K)$ insertions and deletions of the \diu{} problem, with $O(K)$ extra probes to determine what the old value was and to record the new value.

To answer $\texttt{query}(\mathbf{j})$, we first insert intervals that correspond to associated segments of the $(j_i+1)$-th number to the last number in sequence $i$ for $i\in [K]$, to set everything we are not querying to be ``in the union'', no matter how much they were covered by intervals before. That is, we insert $[((i-1)B+j_i)\cdot p, iB\cdot p]$ for each sequence $i$. Then we do one query, which will return the sum of all numbers as if each number not in the query was set to $p$ (or $0$ modulo $p$). This number modulo $p$ is exactly the answer we want. At last, we do $K$ deletions to remove the temporary intervals we inserted earlier for this query, and return the answer. Therefore, every \verb+query+ of \bps{} can be implemented using $O(K)$ insertions and deletions, and one query of the \diu{}. 

Thus, any sequence of $O(B)$ updates and $O(B)$ queries of \bps{} can be implemented using $O(KB)=O(n^{1-\epsilon})$ insertions and deletions, $O(B)=O(n^{\epsilon})$ queries of \diu{}, and extra $O(K)=O(n^{1-\epsilon})$ probes. However, by Theorem~\ref{thmbps}, there is a hard distribution consisting of $O(B)$ updates and $O(B)$ queries, which requires $\Omega(KB\log^2 B/(w+\log B))=\Omega(\epsilon^2 n^{1-\epsilon}\log n)$ probes in expectation in the cell-probe model with word size $w=\Theta(\log n)$. By the above reduction, this hard distribution also induces a distribution over operation sequences for \diu{} with desired number of updates, queries and lower bound on the number of probes. This proves the theorem. 

\end{proof}

By setting $\epsilon=\sqrt{t_u/\log n}$ in Theorem~\ref{thmdiu}, we get the following corollary. 
\begin{corollary}\label{coldiu}
Any \diu{} data structure that answers all queries correctly in a sequence of $O(n)$ operations with probability $\geq 95\%$ with expected amortized insertion and deletion time $t_u$ and query time $t_q$ must have \[
	t_q\geq t_un^{1-\sqrt{t_u/\log n}}.
\]

In particular, $t_u=o(\log n)$ implies that $t_q=n^{1-o(1)}$. 
\end{corollary}

In the following, we are going to prove a lower bound for the \bps{} problem. We will first specify a hard distribution over operation sequences. Then by Yao's Minimax Principle~\cite{Yao77}, it suffices to show that any \emph{deterministic} data structure that answers all queries correctly with high probability under this input distribution must be inefficient. To show this, we will consider a collection of communication games corresponding to different parts of the operation sequence (different time periods). For each communication game, if the data structure is \emph{fast} within the corresponding time period under certain measure of efficiency, then the game can be solved with \emph{low communication cost}. On the other hand, we will prove communication lower bounds for all these games. Summing these lower bounds up, we will be able to prove that the entire execution cannot be efficient. 

\paragraph*{Hard distribution $\mathcal{D}$:}

Without loss of generality, assume $B$ is a power of $2$, and $B=2^b$. The operation sequence will always have $B$ updates and $B$ queries occurring alternatively: $U_0,Q_0,\ldots, U_{B-1},Q_{B-1}$, where the $U_t$'s are updates, and the $Q_t$'s are queries. The operations are indexed by integers between $0$ and $B-1$, or they can be viewed as being indexed by $b$-bit binary strings (which corresponds to the binary representation of the integer). We may use either of these two views in the rest of the paper without further mention. Let $t$ be a binary string, $\rev(t)$ be the string with $t$'s bits reversed. For each $U_t$, we set it to \verb+update(+$\mathbf{j}$\verb+, +$\mathbf{v}$\verb+)+ with $j_i=\rev(t)$ and $v_i$ independently uniformly chosen from $\mathbb{F}_p$ for every $i$. For each $Q_t$, we set it to \verb+query(+$\mathbf{j}$\verb+)+ with $j_i$ independently and uniformly chosen from $[B]$. Different operations are sampled independently. Indicate this distribution by $\mathcal{D}$. 

\paragraph*{Communication game:}

The hard distribution $\mathcal{D}$ has a fixed pattern of updates and queries. Let us fix two \emph{consecutive} intervals $I_A, I_B$ of operations (with $I_A$ before $I_B$) in the sequence. Define the communication game $G(I_A,I_B)$ between two players Alice and Bob to be the following: sample a sequence from $\mathcal{D}$, Alice is given all operations except for those in $I_B$, Bob is given all operations except for those in $I_A$, their goal is to cooperatively compute the answers to all queries in $I_B$ by sending messages to each other alternatively.

Given a deterministic \bps{} data structure, a nature way to solve this game is to let Bob first simulate the data structure up to the beginning of $I_A$, then skip all the operations in $I_A$ and try to continue simulating on $I_B$. Every time Bob needs to probe a cell, if it was probed in $I_A$, he sends a message to Alice asking for the new value, otherwise he knows exactly what is in the cell from his own simulation. The challenge for Bob is to figure out which cells were probed. Our main idea is to introduce a prover Merlin, who knows both Alice and Bob's inputs. Merlin will tell them this information in a \emph{unique} and \emph{succinct} way. Moreover, the players will be able to \emph{verify} whether the message from Merlin is exactly what they expect, and will be able to solve the task efficiently if it is. This motivates the following definition of a new communication model.

\paragraph*{Communication model $\mathcal{M}$:}

Draw an input pair $(x, y)$ from a known distribution. Alice is given $x$, Bob is given $y$ and Merlin is given both $x$ and $y$. Their goal is to compute some function $f(x, y)$. As part of the communication protocol, the players must specify a \emph{unique} string $Z(x, y)$ for every possible input pair, which is the message Merlin is supposed to send. A communication procedure shall consist of the following four stages:

\begin{enumerate}
	\item
		Merlin sends a message $z$ to Alice and Bob based on his input pair $(x, y)$; 
	\item
		Alice and Bob communicate based on $(x, z)$ and $(y, z)$ as in the classic communication setting with public randomness. Merlin does not see the random bits when sending the message $z$;
	\item
		Alice and Bob decide to accept or reject;
	\item
		if the players accept in Stage 3, they return a value $v$.
\end{enumerate}

In this model, we say a protocol computes function $f$ with error $\varepsilon$ and communication cost $C$, if 
\begin{enumerate}
	\item
		Alice and Bob accept in Stage 3 if and only if Merlin sends what he is supposed to send, i.e., $z=Z(x, y)$ (with probability 1),
	\item
		given $z=Z(x, y)$, the value $v$ they return equals to $f(x, y)$ with probability $\geq 1-\varepsilon$ over the input distribution and public randomness,
	\item
		given $z=Z(x, y)$, the expected number of bits communicated between the three players in Stage 1 and 2 is no more than $C$ over the input distribution and public randomness.
\end{enumerate}

\begin{remark}
The public randomness used in Stage 2 does not help the players in general. Nevertheless, we still keep it in the definition for the sake of neatness of our proof. 
\end{remark}

$Z(x, y)$ can be viewed as a question that the players want to ask Merlin about their joint input. One can potentially design more efficient protocols in this model than in the classic communication model if verifying the answer to this question is easier than computing it. 

With respect to this communication model, on one hand, we can show that given a ``good'' \bps{} data structure, we can solve the communication game efficiently (Lemma~\ref{commupper}). On the other hand, we prove a communication lower bound for it (Lemma~\ref{commlower}). Combining these two lemmas, we conclude that no ``good'' data structure exists. 

\begin{lemma}\label{commupper}
Given a deterministic \bps{} data structure for the cell-probe model with word size $w$ that is correct on all queries in a random operation sequence drawn from $\mathcal{D}$ with probability $\geq 90\%$, we can solve the communication game $G(I_A,I_B)$ with error $0.1$ and cost \[
O\left(\E_{\mathcal{D}}\left[|P_A|+|P_B|+|P_A\cap P_B|\cdot w\right]\right)\] in \mm{}, where $P_A$ ($P_B$ resp.) is the (random) set of cells probed in time period $I_A$ ($I_B$ resp.) by the data structure. 
\end{lemma}

\begin{proof}

We prove the lemma by showing the following protocol is efficient in terms of $P_A$ and $P_B$. 

\paragraph*{Communication protocol:}

\begin{enumerate}[Step 1:]
	\setcounter{enumi}{0}
	\item
		(Merlin sends the key information)
		
		Merlin first simulates the data structure up to the beginning of $I_B$, which is also the end of $I_A$, and records the set $P_A$, all the cells that were probed in time period $I_A$. Then Merlin continues simulating the operations in $I_B$. At the meanwhile, every time he probes a memory cell, he checks if this cell has been probed in $I_B$ before and checks if it was probed in $I_A$ (in set $P_A$). If this is the first time probing this cell since the beginning of $I_B$, Merlin will send one bit to Alice and Bob. He sends ``1'' if the cell was probed in $I_A$, and sends ``0'' otherwise.
		
	\item
		(Alice and Bob simulate the data structure to accomplish the task)
		
		Alice simulates the data structure up to the beginning of $I_B$, and records $P_A$. Since Bob does not have any information about operations in $I_A$, he instead simulates up to the beginning of $I_A$, then tries to skip $I_A$ and simulate the operations in $I_B$ directly. Of course, the memory state Bob holds might be very different from what it should look like at the beginning of $I_B$. But with the help of Merlin's message, Bob will be able to figure out the difference. 
		
		As Bob simulates the data structure, every he needs to probe a cell, he first checks if this is the first time probing this cell since the beginning of $I_B$. If it is not, Bob knows its content from the last probe. Otherwise, he looks at the next bit of Merlin's message. If it is a ``0'', Merlin is claiming that this cell was not probed in $I_A$. Thus, its content has not been changed since the beginning of $I_A$. Bob has the information in his own copy of memory. If it is a ``1'', Bob sends the address of this cell to Alice, Alice will send back its content. At the same time, Alice checks if the cell was actually probed in $I_A$. If the check fails, they report ``Merlin is cheating'' (they reject), and abort the protocol. At last, Bob updates this cell in his own copy of memory, and records that it has been probed in $I_B$. 
		
		If there are no more bits left in Merlin's message when Bob needs to look at the next bit, or there are still unread bits when Bob has finished the simulation, the players reject. 
		
	\item
		(Players verify that Merlin is truthful)
		
		According to the simulation in Step 2, Alice takes the set $P_A$. Bob generates the set of cells that were probed in $I_B$ but Merlin claims that they were not in $P_A$ (and thus did not ask Alice for the contents). They check if these two sets of cells are disjoint (all cells that Merlin claims not probed in $I_A$ are actually not) using the \emph{zero-error} sparse set disjointness protocol in Theorem~\ref{thmsparse} of H\aa{}stad and Wigderson. If the two sets intersect, they report ``Merlin is cheating'' (reject), otherwise they report ``Merlin is truthful'' (accept) and Bob returns the answers he has computed for all queries in $I_B$. 
		
\end{enumerate}

Step 1 above describes what Merlin is \emph{supposed} to do, and thus defines $Z(x, y)$. 
\iftoggle{conf}{
We can show that the protocol is a valid protocol in \mm{}, and solves the communication game efficiently. The detailed proof can be found in the full version.
}{
The following shows that the above protocol is a valid protocol in \mm{}, and solves the communication game efficiently. 

\begin{enumerate}
\item
	If Merlin tells the truth ($z=Z(x, y)$), it is not hard to see that the players will always accept. If $z$ is a prefix of $Z(x, y)$ or $Z(x, y)$ is a prefix of $z$, Bob will detect it in Step 2 and reject. Otherwise let the $i$-th bit be the first bit that $z$ and $Z(x, y)$ differ. As the first $i-1$ bits are the same in $z$ and $Z(x, y)$, Bob will correctly simulate the data structure up to that point, right before a probe that causes Bob to read the $i$-th bit of $z$. Thus the cell probed by the data structure corresponding to the $i$-th bit will be the same in Bob's simulation and in the actual execution.  If $Z_i(x, y)=0,z_i=1$, the cell is not probed in $I_A$ but Merlin claims it is. The players can detect the mistake and will reject in Step 2. If $Z_i(x, y)=1,z_i=0$, Merlin claims the cell is not probed in $I_A$ but it is. In this case, the cell will belong to both Alice's and Bob's sets in Step 3. Therefore, the sparse set disjointness protocol will return ``intersect''. The players will reject. This proves that Alice and Bob accept if and only if $z=Z(x, y)$. 
\item
	Given $z=Z(x, y)$, Bob will successfully simulate the data structure on all operations in $I_B$. As the data structure correctly answers all queries simultaneously with $\geq 90\%$ probability, in particular, it answers all queries in $I_B$ correctly with $\geq 90\%$ probability. Thus, the error probability is no more than $0.1$. 
\item
	Given $z=Z(x, y)$, Merlin sends exactly one bit for each cell in $P_B$, $|z|=|P_B|$. In Step 2, the players send $O(w)$ bits for every ``1'' in $z$, which is $O(|P_A\cap P_B|\cdot w)$ in total. In Step 3, by Theorem~\ref{thmsparse}, the players send $O(|P_A|+|P_B|)$ bits in expectation to compute sparse set disjointness. in expectation over the randomness of the protocol and the input distribution $\mathcal{D}$, the protocol uses 
	\[O\left(\E_{\mathcal{D}}\left[|P_A|+|P_B|+|P_A\cap P_B|\cdot w\right]\right)\] bits of communication as we claimed. 
\end{enumerate}

This proves the lemma.
}

\end{proof}

Let $s$ be a binary string of length less than $b=\log B$. Define $I_s$ to be the interval consisting operations $\{U_t, Q_t: \textrm{$s$ is a prefix of $t$}\}$. Let $G(s)=G(I_{s0},I_{s1})$ be the communication game defined by $I_{s0}$ and $I_{s1}$, e.g., in game $G(\emptyset)$, Alice receives all operations in the first half of the sequence as her input, and Bob receives the second half, in game $G(0)$, Alice receives the first quarter and the second half, Bob receives the second quarter and the second half. 
\begin{lemma}\label{commlower}
For $p\geq B$ large enough, the communication game $G(s)$ requires communication cost at least $\Omega(2^{-|s|}KB(b-|s|))$ for any protocol with error $0.1$ in \mm{}, where $|s|$ is the length of string $s$. 
\end{lemma}

We will defer the proof of Lemma~\ref{commlower} to Section~\ref{sectcomm}.
\iftoggle{conf}{Using these two lemmas, we will be able to prove Theorem~\ref{thmbps}. As the proof is similar to~\cite{PD04a,PT11}, we omit it in the conference version. See the full version for more details.}
{Using these two lemmas, we are ready to prove our data structure lower bound.

\begin{proof}[Proof of Theorem~\ref{thmbps}]
Fix a (randomized) data structure for \bps{} problem, which errors with probability no more than $0.05$, and in expectation, probes $T$ cells on an operation sequence drawn from $\mathcal{D}$. By Markov's inequality and union bound, there is a way to fix the random bits used by the data structure, such that the error probability is no more than $0.1$, and probes at most $3T$ cells in expectation. In the following, we show that for such deterministic data structure, $T$ must be large.

For binary string $s$ of length no more than $\log B$, define $P_s$ to be the (random) set of cells probed by the data structure in $I_s$. For every $s$, Lemma~\ref{commupper} and Lemma~\ref{commlower} together implies that \[
\E_{\mathcal{D}}\left[|P_{s0}|+|P_{s1}|+|P_{s0}\cap P_{s1}|\cdot w\right]\geq \Omega\left(2^{-|s|}KB(b-|s|)\right).
\]

Now sum up the two sides over all binary strings $s$ of length at most $b-1$. For the left-hand-side, fix an operation sequence. In the sum $\sum_{s} (|P_{s0}|+|P_{s1}|)$, every probe will be counted at most $\log B$ times, because the probes during $U_t$ or $Q_t$ will be counted only when $s$ is a prefix of $t$. In the sum $\sum_s |P_{s0}\cap P_{s1}|$, for each cell in $|P_{s0}\cap P_{s1}|$, we refer it to its first probe in $I_{s1}$. Every probed will be referred to at most once: consider a probe during $U_t$ or $Q_t$, assume the last probe to this cell happened during $U_{t'}$ or $Q_{t'}$ for some $t'<t$, this probe will be referred only when $s0$ is a prefix of $t'$ and $s1$ is a prefix $t$, i.e., $s$ is the longest common prefix of $t'$ and $t$. Therefore, the left-hand-side sums up to at most $3T\cdot (w+\log B)$. The right-hand-side sums up to
\[
	\sum_{s:|s|<b} \Omega(KB2^{-|s|}(b-|s|))=\sum_{|s|=0}^{b-1}\Omega(KB(b-|s|))=\Omega(KB\log^2 B)
\]

This implies $T\geq \Omega(KB\log^2 B/(w+\log B))$, which proves the theorem.
\end{proof}
}

\section{Communication Lower Bound}\label{sectcomm}

Before proving the communication lower bound for the game $G(s)$ itself, we first do a ``clean-up'' to make the problem more accessible. In particular, we show that the communication game we want to prove a lower bound for is essentially the \emph{\mi{}} problem. 

In the \mi{} problem, Alice is given a vector $x\in \mathbb{F}_p^{LK}$. Bob is given an $L$-tuple of vectors $y=(y_1,y_2,\ldots,y_L)$, such that $y_i\in \mathbb{F}_p^{LK}$ for each $i\in[L]$. Moreover, if we divide the $LK$ coordinates into $K$ blocks of $L$ coordinates each in the most natural way (first block is the first $L$ coordinates, second block is the next $L$ coordinates, etc), each $y_i$ will be a $\{0,1\}$-vector with \emph{at most} one $1$ in each block. Their goal is to compute the $L$ inner products over $\mathbb{F}_p$: $f(x, y)=(\left<x,y_1\right>,\left<x,y_2\right>,\ldots,\left<x,y_L\right>)$. In other words, Alice gets an array, Bob gets $L$ sets of indices of the array (of some restricted form). They want to figure out together for each set, what is the sum of elements in the corresponding entries. 

\iftoggle{conf}{
We claim that the communication game $G(s)$ is equivalent to the \mi{} problem with the same parameters $K$ and $p$ as in $G(s)$, parameter $L=B/2^{|s|+1}$ and input distribution $\mu$. Moreover, $\mu$ is ``close'' to the uniform distribution. The detailed proof of equivalence can be found in the full version.

}{
Fix a communication game $G(s)$ defined by $I_A=I_{s0}$ and $I_B=I_{s1}$. By the way we set up the hard distribution $\mathcal{D}$, every update operation will always update the set of entries, only the values change. Therefore, the only thing about the sequences Bob does not know is the values in the entries that are updated in $I_A$. As Bob knows the values in all other entries right before each query, the players' actual goal is to figure out the prefix sums as if there were only those unknown entries, which can be formulated as an instance of the \mi{} problem. Moreover, input distribution $\mathcal{D}$ will induce an input distribution for the \mi{} problem. We just need to prove a communication lower bound under that distribution. 

More specifically, let $L=B/2^{|s|+1}=|I_A|=|I_B|$. Define the following function $F$ which maps a sequence of operations $(U_0,Q_0,\ldots,U_{B-1},Q_{B-1})$, which is a possible outcome of $\mathcal{D}$, to an instance of the \mi{} problem. Consider all updates in interval $I_A=I_{s0}$, let $\mathcal{E}_k$ be the set of entries of sequence $k$ which are updated in $I_A$. To get a \mi{} instance, we set the $l$-th coordinate in $k$-th block of Alice's input $x$ to be the sum of values in $l$ first entries in $\mathcal{E}_k$ (the $l$ entries with smallest indices), for $l\in [L]$ and $k\in [K]$. For Bob's input $y_i$, consider the $i$-th query in $I_B$, let it be $\texttt{query}(j_{i,1},\ldots,j_{i,K})$, querying the sum of first $j_{i,k}$ numbers in sequence $k$. For $k\in [K]$, assume there are $l$ entries in $\mathcal{E}_k$ with indices at most $j_{i,k}$, which will be summands in the $i$-th query. If $l=0$, we set all coordinates in the $k$-th block of $y_i$ to $0$, otherwise, we set the $l$-th coordinate in the block to $1$. This defines the function $F$. It is not hard to see that the inner product $\left<x,y_i\right>$ encodes the dependence of $i$-th query in $I_B$ on updates in $I_A$. Moreover, it is easy to verify that under the mapping of $F$, the distribution over operation sequences $\mathcal{D}$ induces a distribution over the input pairs $(x, y)$ for \mi{}, with some probability measure $\mu$, such that

\begin{enumerate}
	\item
		$x$ and $y$ are independent under $\mu$, i.e., $\mu=\mu_x\times \mu_y$ is a product distribution;
	\item
		$\mu_x$ is the uniform distribution over $\mathbb{F}_p^{LK}$;
	\item
		$\mu_y$ is close to being uniform: all $K$ blocks in all $L$ vectors in $y$ are independent, and in each block, each one of the $L+1$ possibilities will occur with probability no more than $1/L$, as adjacent elements in $\mathcal{E}_k$ are spaced by exactly $B/L-1$ numbers. In particular, $\mu_y(\{y\})\leq L^{-LK}$ for any singleton. 
\end{enumerate}

In the following, we will only use the above three properties of $\mu$ in the proofs. }
In this setting, Alice's input carries $O(LK\log p)$ bits of information. Bob's input carries $O(LK\log L)$ bits of information. The following lemma shows that the best strategy is just to let one of the players send the whole input to the other, even with Merlin's help in \mm{}.

\begin{lemma}\label{lowermi}
For large enough $p$ and $L$, solving the \mi{} problem in \mm{} with error $0.15$ requires communication cost $\Omega\left(\min\{LK\log p, LK\log L\}\right)$ under input distribution $\mu$.
\end{lemma}

\iftoggle{conf}{
}{
Before proving this lemma, we first show that it implies Lemma~\ref{commlower}. 

\begin{proof}[Proof of Lemma~\ref{commlower}]

Fix a protocol $P$ for $G(s)$ with error $0.1$ and cost $C$. We are going to use it to solve \mi{}. Let us first assume that there is a sequence of public random bits that all three parties can see. We will first design a protocol in this setting, then try to get rid of this extra requirement by fixing the random bits.

For an input pair $(x,y)\sim \mu$ for the \mi{} problem, consider the following protocol:
\paragraph*{Preprocessing:} Use the public randomness to sample an operation sequence from $\mathcal{D}$ conditioned on that $F$ maps it to $(x, y)$. It is easy to verify that all operations outside $I_B$ does not dependent on $y$ and all operations outside $I_A$ does not depend on $x$. Therefore, with no communication, all three parties get their inputs for $G(s)$. 
\paragraph*{Simulate $P$:} Alice and Bob run protocol $P$ to compute all answers to queries in $I_B$. 
\paragraph*{Postprocessing:} Bob knows all updates outside $I_A$, and from the value returned by the communication game, he gets to know the answers to all queries. Therefore, Bob can compute for $i$-th query in $I_B$, the sum of all entries updated in $I_A$ that are summands of the query, which is exactly $\left<x,y_i \right>$, by subtracting all other summands from the answer. With no further communication, Bob can figure out the solution to the \mi{} problem.

After Preprocessing, the inputs the players get for $G(s)$ will be distributed as $\mathcal{D}$. Therefore, in Simulate $P$, the communication cost will be $C$ in expectation, and error probability will be $0.1$ over the randomness of $P$, input $(x,y)$ and random bits $r$ used in Preprocessing. By Markov's inequality and union bound, there is a way to fix $r$, such that the communication cost is at most $4C$ in expectation and error probability is at most $0.15$ over the randomness of $P$ and $(x, y)$. To show that the protocol can be implemented in \mm{}, we hardwire $r$ and define for each input pair $(x, y)$, $Z(x, y)$ to be the message Merlin is supposed to send in Simulate $P$ when Preprocessing uses random bits $r$. Alice and Bob accept if and only if they were to accept in $P$. 

It is easy to verify that the above protocol solves the \mi{} problem under input distribution $\mu$ with error $0.15$ and cost $4C$. However, by Lemma~\ref{lowermi}, we have a lower bound of $\Omega(\min\{LK\log p, LK\log L\})$ on the communication cost. Together with $p\geq B\geq L$, we have $C\geq \Omega(LK\log L)=\Omega(2^{-|s|}KB(b-|s|))$. This proves the lemma.
\end{proof}
}

To prove a communication lower bound in \mm{}, the main idea is to use the uniqueness of the certificate ($Z(x, y)$) and the perfect completeness and soundness. To start with, let us first consider the case where the protocol is deterministic, and the communication cost $C$ is defined in worst case instead of in expectation. 

In this case, fix one Merlin's possible message $z$, it defines a communication problem between Alice and Bob in the classic model: check whether $Z(x, y)=z$. By definition, the players can solve this task with zero error. By the classic monochromatic rectangle argument, we can partition the matrix $M(f)$ (defined in \iftoggle{conf}{Appendix~\ref{sectpre_app}}{Section~\ref{sectpre}}) into exponentially in $C$ many combinatorial rectangles, such that in each rectangle, either $Z(x, y)=z$ for every pair or $Z(x, y)\neq z$ for every pair. In particular, it partitions the set $Z^{-1}(z)$ into combinatorial rectangles. Moreover, for each rectangle with $Z(x, y)=z$, the protocol associates it with a value, which is the value returned in Stage 4. Now we go over all possible $z$'s, which is again exponentially in $C$ many. Every input pair belongs to exactly one of the $Z^{-1}(z)$'s. By cutting all $Z^{-1}(z)$, along with their partitioning into rectangles, and pasting into one single matrix, it induces a partitioning of the \emph{whole} matrix $M(f)$ into $2^{O(C)}$ rectangles. The values associated with the rectangles should match the actual function values with high probability. Therefore, there must be large \emph{nearly} monochromatic rectangles in $M(f)$. \iftoggle{conf}{\footnote{See Appendix~\ref{sectpre_app} for definitions.}}{}

A nature final step of the proof, as in many communication complexity lower bound proofs, would be to show that all nearly monochromatic rectangles are small. However, in the \mi{} problem, there do exist large monochromatic rectangles.

Fix a set $S\subseteq [LK]$ of coordinates, such that it has $L/\log L$ coordinates in each block. Let $X=\{x:\forall j\in S, x_j=0\}$, $Y=\{(y_1,\ldots,y_L):\forall j\notin S, i\in [L], y_{i,j}=0\}$. $X\times Y$ is a monochromatic rectangle with value $(0,\ldots, 0)$, and $\mu(X\times Y)=\Theta(p^{-LK/\log L}\cdot (\log L)^{-LK})=\Theta(2^{-LK(\log p/\log L+\log\log L)})$. In particular, for $p=O(L)$, we can only prove lower bounds no better than $\Omega(LK\log\log L)$ using this approach only. 

However, these $y$'s are not what a ``typical'' Bob's input should look like. Since the ones in $\{y_1,\ldots,y_L\}$ only appear in $(1/\log L)$-fraction of the coordinates, while a random input with very high probability should have ones appearing in a constant fraction of the coordinates. This motivates the following definition of \emph{\es{}}:
\begin{definition}
Define Bob's input to be \emph{\es}, if for any set of coordinates of size at most $0.1LK$, the number of ones in all $L$ vectors in Bob's input in these coordinates is no more than $0.9LK$. 
\end{definition}

\iftoggle{conf}{By Hoeffding's inequality and union bound, we can prove that most of the inputs are \es{}. The detailed proof can be found in the full version.
}{}

\begin{lemma}\label{lemES}
	For $L\cdot K$ large enough, with probability $\geq 95\%$, Bob's input is {\es}. 
\end{lemma}

\iftoggle{conf}{}{
\begin{proof}
	Draw a random input $(x,y)$ from $\mu$, and try to upper-bound the probability that it is not {\es}. Fix a set of coordinates $S=S_1\cup S_2\cup \cdots \cup S_K$ of size at most $0.1LK$, where $S_k$ is a subset of coordinates in block $k$. Let $\xi_{ki}$ be the random variable $\left<1_{S_k}, y_i\right>$, which is the number of ones in $y_i$ that is in $S_k$. By the properties of $\mu$, in $k$-th block of $y_i$, each of the $L+1$ possibilities occurs with probability no more than $1/L$, and is independent of all other blocks and vectors. Thus, we have that $\E[\xi_{ki}]\leq |S_k|/L$, $\xi_{ki}\in [0,1]$, and $\xi_{ki}$'s are independent. Let $\xi=\sum_{k=1}^K\sum_{i=1}^L \xi_{ki}$, be the number of ones in all $L$ vectors in $S$. We have $\E[\xi]\leq 0.1LK$. Thus, by Hoeffding's inequality, 
	\[
		\Pr[\xi>0.9LK]\leq \Pr[\xi-E[\xi]>0.8LK]\leq e^{-32LK/25}.
	\]
	
	However, the number of possible set $S$'s is no more than $2^{LK}$. By union bound, the probability that there exists an $S$ violating the constraint is at most $e^{-32LK/25}\cdot 2^{LK}<0.05$, for large enough $L\cdot K$.
\end{proof}
}

Instead of upper-bounding how many input pairs a nearly monochromatic rectangle can contain, we are going to upper-bound the measure of \es{} inputs in it. Define set $E$ to be the all input pairs $(x, y)$ that $y$ is \emph{not} \es{}. By Lemma~\ref{lemES}, $\mu(E)\leq 0.05$. The following lemma shows that there are no large nearly monochromatic rectangles if we ignore all elements in $E$. 

\begin{lemma}\label{almonorect}
For large enough $p, L$, every $0.7$-monochromatic rectangle $R$, which is disjoint from $E$, must have \[\mu(R)\leq \max\{2^{-\Omega(LK\log L)},2^{-\Omega(LK\log p)}\}.\]
\end{lemma}

\iftoggle{conf}{
Using the above lemma, we are ready to prove Lemma~\ref{lowermi}. Both proofs can be found in the full version.
}
{
\begin{proof}
Fix a combinatorial rectangle $R=X\times Y$, such that there is a value $v=(v_1,\ldots,v_L)$ that $\mu(R\cap f^{-1}(v))\geq 0.7 \mu(R)$. We want to prove that $\mu(R)$ must be small. First, without loss of generality, we can assume that for every $y\in Y$, $\mu((X\times \{y\})\cap f^{-1}(v))\geq 0.6 \mu(X\times \{y\})$, i.e., every column in $R$ is $0.6$-monochromatic. Since in general, by Markov's inequality, at least $1/4$ (with respect to $\mu_y$) of the columns in $R$ are $0.6$-monochromatic, we can just apply the following argument to the subrectangle induced by $X$ and these columns, and only lose a factor of $4$. 

Let $|Y|=q, Y=\{y_1,y_2,\ldots,y_q\}$, $y_i=(y_{i1},\ldots,y_{iL})$. Let $r_i$ be the dimension of subspace in $\mathbb{F}_p^{LK}$ spanned by vectors in first $i$ $L$-tuples: $\{y_{jl}:1\leq j\leq i,l\in [L]\}$, and $r_0=0$. Let $r=r_q$ be the dimension of the subspace spanned by all vectors in $Y$. Define $\nu$ to be the probability distribution over $Y$ such that $\nu(y_i)=(r_i-r_{i-1})/r$. $\nu(y_i)$ is proportional to ``the number of new dimensions $y_i$ introduces''. Thus, $\nu$ is supported on no more than $r$ elements in $Y$.

Since every column in $Y$ is 0.6-monochromatic under $\mu_x$, $R$ will also be 0.6-monochromatic under $\mu_x\times \nu$. By Markov's inequality again, at least 1/5 of the rows are 0.5-monochromatic in $R$ under $\mu_x\times \nu$. However, when $r$ is large, there cannot be too many such rows, even in the whole matrix. For some 0.5-monochromatic row $x$, let $S\subseteq Y$ be the set of columns with value $v$ in that row. By definition, we have $\nu(S)\geq 0.5$. By the way we set up the distribution $\nu$, the linear space spanned by all vectors in $S$ must have dimension at least $r/2$. This adds at least $r/2$ independent linear constraints on $x$. There can be at most $p^{-r/2}$-fraction of $x$'s satisfying all linear constraints in the whole matrix. By taking a union bound on all possible $S$'s, we obtain an upper bound on the measure of 0.5-monochromatic rows in $R$ under $\mu_x\times \nu$. More formally, we have

\[
	\begin{aligned}
		\frac{1}{5}\mu_x(X)&\leq \mu_x\left(\{x:\nu(\{y:f(x,y)=v\})\geq 0.5\}\right) \\
		&=\Pr_{x\sim \mu_x}\left[\Pr_{y=(y_1,\ldots,y_L)\sim \nu}\left[\forall l\in [L],\left<x,y_l\right>=v_l\right]\geq 0.5\right]\\
		&=\Pr_{x\sim \mu_x}\left[\exists S\subseteq \mathrm{supp}(\nu),\nu(S)\geq 0.5,\forall y\in S,\forall l\in [L],\left<x,y_l\right>=v_l\right]\\
		&\leq \sum_{\stackrel{S\subseteq \mathrm{supp}(\nu)}{\nu(S)\geq 0.5}}\Pr_{x\sim \mu_x}\left[\forall y\in S,\forall l\in L,\left<x,y_l\right>=v_l\right] \\
		&\leq \sum_{\stackrel{S\subseteq \mathrm{supp}(\nu)}{\nu(S)\geq 0.5}}p^{-r/2} \\
		&\leq 2^r p^{-r/2}\leq 2^{-\Omega(r\log p)}.
	\end{aligned}
\]

Therefore, if $r\geq 0.1LK$, we have $\mu(R)\leq \mu_x(X)\leq 2^{-\Omega(r\log p)}\leq 2^{-\Omega(LK\log p)}$. 

Otherwise, $r\leq 0.1LK$. In this case, we are going to show that, it is impossible to pack too many vectors the the form of Bob's inputs into any subspace of small dimension. In particular, we will upper bound $\mu_y(Y)$, the measure of Bob's \es{} inputs, when their span has dimension $r$. Fix a basis of the span of vectors in $Y$, consisting of $r$ vectors in $\mathbb{F}_p^{LK}$.\footnote{These vectors do not have to be from Bob's inputs.} Without loss of generality, we can assume that for each basis vector, there is a coordinate in which this vector has value 1, and all other basis vectors have value 0, because we can always run a standard Gaussian elimination to transform the basis into this form. Let $T$ be this set of $r$ coordinates. As each vector in the subspace is a linear combination of the basis, fixing the values in coordinates in $T$ uniquely determines a vector in the subspace. By definition of \es{}, and $|T|=r\leq 0.1LK$, each $L$-tuple $y_i\in Y$ can have at most $0.9LK$ 1's in $T$. For all $LK$ blocks in the $L$ vectors, there are ${LK\choose \leq 0.9LK}$ choices for the set of blocks with ones. Moreover, each $y_i$ can have at most one 1 in each block. If a vector has a 1 in a block, there will be at most $L$ different choices to place the 1. Otherwise, we must set all coordinates in $T$ in that block to be $0$. After fixing values all coordinates in $T$, there can be at most one such vector in the subspace with matching values. Thus, we have
\[
	\begin{aligned}
	\mu_y(Y)&\leq {LK\choose \leq 0.9LK}\cdot L^{0.9LK}\cdot L^{-LK}\\
	&\leq L^{-0.1LK}\cdot 2^{LK}\leq 2^{-\Omega(LK\log L)}.
	\end{aligned}
\]

Therefore, in this case, we have $\mu(R)\leq \mu_y(Y)\leq 2^{-\Omega(LK\log L)}$. Combining both cases, we conclude that $\mu(R)\leq \max\{2^{-\Omega(LK\log L)},2^{-\Omega(LK\log p)}\}$. 
\end{proof}

Using the above lemma, we can prove communication lower bound for \mi{}.

\begin{proof}[Proof of Lemma~\ref{lowermi}]
Fix a protocol that solves the \mi{} problem with error 0.15 and cost $C$ in \mm{}. As we are working with a fixed input distribution, randomness in the protocol shall not help. In particular, by Markov's inequality and union bound, there is a way to fix the internal (public) randomness of the protocol, such that the success probability is at least $80\%$ and communication cost no more than $5C$. From now on, let us assume the protocol is deterministic, and success probability is at least $80\%$. 

For some message $z$ sent by Merlin, let $T_z=Z^{-1}(z)$ be set of the input pairs $(x,y)$ such that $z$ is message Merlin is supposed to send when the players get these input pairs. As Alice and Bob are always able to decide whether $z=Z(x, y)$, the classical combinatorial rectangle argument induces a way to partition each $T_z$ into rectangles based on the transcript between Alice and Bob. Moreover, the set $\{T_z\}_{z\in\{0,1\}^*}$ induces a partition of all possible input pairs. Therefore, provided that Merlin tells the truth, the entire transcript $\gamma(x, y)$, which includes both Merlin's message and the transcript between Alice and Bob, induces a partition of the matrix $M(f)$ into combinatorial rectangles $\{R_i\}$. For each $R_i$ with transcript $\gamma_i$, the players will return a fixed answer $v_i$ for every pair of inputs in the rectangle. 

By the definition of communication cost, we have $\E_{(x,y)\sim \mu}[|\gamma(x,y)|]\leq 5C$. Thus, by Markov's inequality
\begin{equation}\label{eqn1}
	\sum_{R_i:|\gamma_i|\leq 50C}\mu\left(R_i\setminus E\right)\geq \sum_{R_i:|\gamma_i|\leq 50C}\mu\left(R_i\right)-\mu(E)>0.9-0.05=0.85
\end{equation}

By the definition of probability of computing $f$ correctly, \[
\sum_{R_i}\mu(R_i\cap f^{-1}(v_i))\geq 0.8.
\]

By Markov's inequality, 
\[
\begin{aligned}
	&\sum_{R_i:\mu((R_i\setminus E)\cap f^{-1}(v_i))<0.7\mu(R_i\setminus E)}\mu(R_i\setminus E) \\
	&= \sum_{R_i:\mu(R_i\setminus E\setminus f^{-1}(v_i))>0.3\mu(R_i\setminus E)}\mu(R_i\setminus E) \\
	&\leq \frac{1}{0.3}\sum_{R_i:\mu(R_i\setminus E\setminus f^{-1}(v_i))>0.3\mu(R_i\setminus E)}\mu(R_i\setminus E\setminus f^{-1}(v_i)) \\
	&\leq \frac{1}{0.3}\sum_{R_i}\mu(R_i\setminus f^{-1}(v_i)) \\
	&=\frac{1}{0.3}\left(1-\sum_{R_i}\mu(R_i\cap f^{-1}(v_i))\right)\leq 2/3
\end{aligned}
\]

Thus, we have
\begin{equation}\label{eqn2}
	\sum_{R_i:\mu((R_i\setminus E)\cap f^{-1}(v_i))\geq 0.7\mu(R_i\setminus E)}\mu(R_i\setminus E)\geq 1-\mu(E)-2/3>0.25
\end{equation}

Let $\mathcal{R}$ be the (disjoint) union of all $R_i$ with value $v_i$ and transcript $\gamma_i$, such that 
\[
	\mu((R_i\setminus E)\cap f^{-1}(v_i))\geq 0.7\mu(R_i\setminus E)
\]
and $|\gamma_i|\leq 50C$. By (\ref{eqn1}) and (\ref{eqn2}), we have $\mu(\mathcal{R}\setminus E)\geq 0.1$. However, there can be only $2^{O(C)}$ different such $\gamma_i$'s, and thus $2^{O(C)}$ such rectangles. There must be some transcript $\gamma_i$ such that $\mu(R_i\setminus E)\geq 2^{-O(C)}$ and $R_i\setminus E$ is $0.7$-monochromatic rectangle under distribution $\mu$. Therefore, by Lemma~\ref{almonorect}, we have $C\geq \Omega(\min\{LK\log L, LK\log p\})$, which proves the lemma. 

\end{proof}

\section{Applications to Dynamic Graph Problems}\label{sectapp}
In this section, we present three applications of our main theorem to dynamic graph problems. See Appendix~\ref{sectcata} for formal definitions of the problems.

\begin{corollary}\label{cordyngraph}
For the following three dynamic graph problems:
\begin{enumerate}[(a)]
	\item
		dynamic \#SCC,
	\item
		dynamic planar s-t min-cost flow,
	\item
		dynamic weight s-t shortest path,
\end{enumerate}
any data structure with amortized expected update time $o(\log n)$, and error probability $\leq 5\%$ under polynomially many operations must have amortized expected query time $n^{1-o(1)}$. 
\end{corollary}

\begin{proof}(sketch)
To prove the lower bounds, we are going to give three reductions from \diu{}. The corollary follows from Corollary~\ref{coldiu}. 
\begin{enumerate}[(a)]
	\item To solve \diu{}, we maintain the following graph $G$: $G$ has a Hamiltonian path $0\rightarrow 1\rightarrow \cdots \rightarrow n$; for every $[a, b]\in \mathcal{I}$, $G$ has an edge $b\rightarrow a$. It is not hard to see that the graph can be maintained efficiently given a dynamic \#SCC data structure, and the total length of the union of $\mathcal{I}$ is exactly $n+1$ minus the number of strongly connected components in $G$. See Figure~\ref{figscc}. 
	
\begin{figure}
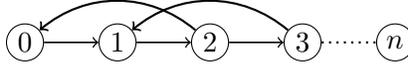

\centering
\tikz[node distance=35pt]{
	\tikzstyle{vertex}=[circle,draw,minimum size=14pt, inner sep=0];
	\node[vertex] (v0) {$0$};
	\node[vertex, right of=v0] (v1) {$1$};
	\node[vertex, right of=v1] (v2) {$2$};
	\node[vertex, right of=v2] (v3) {$3$};
	\node[vertex, right of=v3] (v4) {$n$};
	\draw[->,semithick] (v0) -> (v1);
	\draw[->,semithick] (v1) -> (v2);
	\draw[->,semithick] (v2) -> (v3);
	\draw[dotted, thick] (v3) -- (v4);
	\draw[->,thick] (v2) edge [bend right=40] (v0);
	\draw[->,thick] (v3) edge [bend right=40] (v1);
}
\caption{dynamic \#SCC}\label{figscc}
\end{figure}

	\item
	The underlying graph is shown as Figure~\ref{figflow}. The edges connecting vertices $i$ and $i+1$ have infinite capacities and zero cost. The edges connect $i$ and $i'$ have capacities 1 and cost -1. The edges connecting $i'$ and $i+1$ have capacities 1 and cost 0. All edges connecting to $s$ or $t$ have cost 0. The only values that may change are the capacities of edges connecting to $s$ or $t$. More specifically, we maintain the graph such that the capacity from $s$ to vertex $i$ always equals to the number of intervals in $\mathcal{I}$ with left endpoint $i$, the capacity from $i$ to $t$ always equals to the number of intervals with right endpoint $i$. It is easy to verify that given a dynamic planar s-t min-cost flow data structure, we can efficiently maintain these invariants. To query the total length of the union of $\mathcal{I}$, we query the min-cost flow in $G$ with flow value $|\mathcal{I}|$. For each $i$, the amount of flow from $i$ to $i+1$ (also counting flow going through $i'$) is exactly the number of intervals containing $[i, i+1]$. As $i\rightarrow i'\rightarrow i+1$ has a smaller cost than going to $i+1$ directly from $i$, min-cost flow will prefer to use the former path. Each $[i,i+1]$ contained in any intervals in $\mathcal{I}$ adds a cost of $-1$ to the flow. Therefore, the negate of the cost is exactly the length of the union. 
	
\begin{figure}
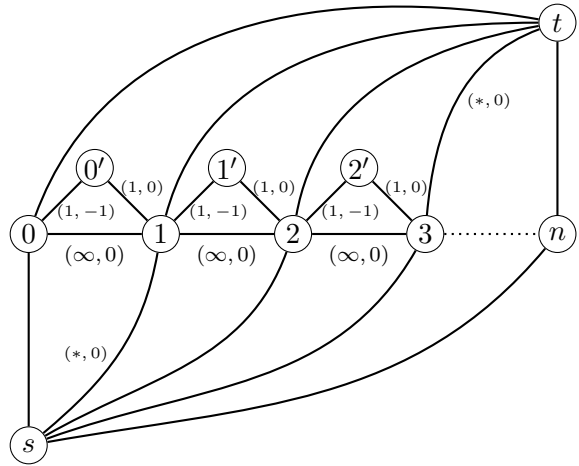

\centering
\tikz[node distance=50pt]{
	\tikzstyle{vertex}=[circle,draw,minimum size=14pt, inner sep=0];
	\node[vertex] (v0) {$0$};
	\node[vertex, right of=v0] (v1) {$1$};
	\node[vertex, right of=v1] (v2) {$2$};
	\node[vertex, right of=v2] (v3) {$3$};
	\node[vertex, right of=v3] (v4) {$n$};
	\node[vertex, below of=v0, node distance=80pt] (vs) {$s$};
	\node[vertex, above of=v4, node distance=80pt] (vt) {$t$};
	\node[vertex, above of=v0, xshift=25pt, node distance=25pt] (v0') {$0'$};
	\node[vertex, above of=v1, xshift=25pt, node distance=25pt] (v1') {$1'$};
	\node[vertex, above of=v2, xshift=25pt, node distance=25pt] (v2') {$2'$};
	\draw[thick] (vs) -- (v0) -- 
	node[below] {\scriptsize $(\infty,0)$} (v1) -- 
	node[below] {\scriptsize $(\infty,0)$} (v2) -- 
	node[below] {\scriptsize $(\infty,0)$} (v3);
	\draw[thick] (v4) -- (vt);
	\draw[thick] (v0) --
	node[right,xshift=-6pt, yshift=-4pt] {\tiny $(1, -1)$} (v0') -- 
	node[right,xshift=-7pt, yshift=5pt] {\tiny $(1,0)$} (v1) -- 
	node[right,xshift=-6pt, yshift=-4pt] {\tiny $(1, -1)$} (v1') -- 
	node[right,xshift=-7pt, yshift=5pt] {\tiny $(1,0)$} (v2) -- 
	node[right,xshift=-6pt, yshift=-4pt] {\tiny $(1, -1)$} (v2') -- 
	node[right,xshift=-7pt, yshift=5pt] {\tiny $(1,0)$} (v3);
	\draw[dotted, thick] (v3) -- (v4);
	\path (v0) edge [thick, out=70, in=170] (vt);
	\path (v1) edge [thick, out=75, in=180] (vt);
	\path (v2) edge [thick, out=80, in=190] (vt);
	\path (v3) edge [thick, out=85, in=200] node[right] {\tiny $(*, 0)$} (vt);
	\path (vs) edge [thick, out=40, in=-100] node[left] {\tiny $(*, 0)$} (v1);
	\path (vs) edge [thick, out=30, in=-110] (v2);
	\path (vs) edge [thick, out=20, in=-120] (v3);
	\path (vs) edge [thick, out=10, in=-130] (v4);
}
\caption{dynamic planar s-t min-cost flow}\label{figflow}
\end{figure}	
	
	\item 
	
	We maintain a graph $G$ such that there is an edge from $s$ to $0$ with weight 0, an edge from $n$ to $t$ with weight $0$, edges from $i$ to $i+1$ with weight 1, and edges from $i+1$ to $i$ with weight $0$. Moreover, for each interval $[a, b]\in\mathcal{I}$, the graph has an edge $a\rightarrow b$ with weight 0. The shortest path from $s$ to $t$ is exactly the length of the union, because for $[i, i+1]$ contained in any interval in $\mathcal{I}$, we can go from $i$ to $i+1$ with zero cost: go to the left endpoint of the interval, then go to the right endpoint using one edge, and go to $i+1$. See Figure~\ref{figsp}. 
	
\begin{figure}
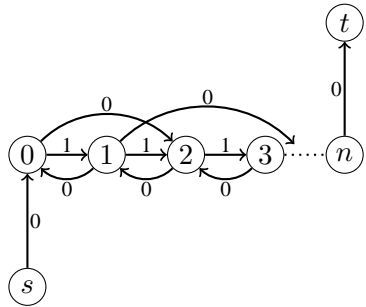

\centering
\tikz[node distance=30pt]{
	\tikzstyle{vertex}=[circle,draw,minimum size=14pt, inner sep=0];
	\node[vertex] (v0) {$0$};
	\node[vertex, right of=v0] (v1) {$1$};
	\node[vertex, right of=v1] (v2) {$2$};
	\node[vertex, right of=v2] (v3) {$3$};
	\node[vertex, right of=v3] (v4) {$n$};
	\node[vertex, below of=v0, node distance=50pt] (vs) {$s$};
	\node[vertex, above of=v4, node distance=50pt] (vt) {$t$};
	\draw [thick,->] (vs) -> node[right,xshift=-3pt] {\scriptsize 0} (v0);
	\draw [thick,->] (v0) -> node[above,yshift=-3pt] {\scriptsize 1} (v1);
	\draw [thick,->] (v1) -> node[above,yshift=-3pt] {\scriptsize 1} (v2);
	\draw [thick,->] (v2) -> node[above,yshift=-3pt] {\scriptsize 1} (v3);
	\draw [thick,dotted] (v3) -- node (v35){} (v4);
	\draw [thick,->] (v4) -> node[left,xshift=3pt] {\scriptsize 0} (vt);
	\path (v1) edge [draw, thick,->,bend left=45] node[below,yshift=3pt] {\scriptsize 0} (v0);
	\path (v2) edge [draw, thick,->,bend left=45] node[below,yshift=3pt] {\scriptsize 0} (v1);
	\path (v3) edge [draw, thick,->,bend left=45] node[below,yshift=3pt] {\scriptsize 0} (v2);
	\path (v0) edge [draw, thick,->,bend left=45] node[above,yshift=-3pt] {\scriptsize 0} (v2);
	\path (v1) edge [draw, thick,->,bend left=45] node[above,yshift=-3pt] {\scriptsize 0} (v35);
}
\caption{dynamic weight s-t shortest path}\label{figsp}
\end{figure}
\end{enumerate}
\end{proof}
}

\section{Final Remarks}\label{sectremark}
In~\cite{CGL15}, Clifford, Gr\o{}nlund and Larsen mentioned a $\log m \log n$ barrier for dynamic data structure lower bounds, where $m$ is the number of different queries (including parameters), and $n$ is the number operations in the sequence we analyse. In some sense, our main result can also be viewed as a ``$(\log m\log n)$-type'' lower bound, as the query takes only $O(1)$ bits to describe. The way we prove this type of lower bound is very different from~\cite{Larsen12a, Larsen12b, CGL15}. We obtain this kind of the lower bound via reduction. Assume we want to prove a lower bound for problem $A$. We first prove a lower bound for some other problem $B$ with large $m$, and get a high lower bound. Then we find a way to implement updates of $B$ using updates of $A$, and queries of $B$ using updates and queries of $A$, and thus derive a lower bound for $A$. Note that it is important that we implement queries in problem $B$ using \emph{both} updates and queries in the original problem. Because if we only use queries to implement queries, in order to keep all the information in the query of $B$, it has to be decomposed into many queries of $A$. A simple calculation shows that in this case, we cannot break the barrier for problem $A$ unless we have already broken it for problem $B$. However, if we use both updates and queries of $A$, it is possible to ``hide'' information in the updates. A good example is the reduction in Proposition~\ref{redps} in Appendix~\ref{sectredps}. However, using this approach, we still cannot beat $\log m'\log n$, where $m'$ is the number of different \emph{updates} (including parameters). Nevertheless, it gives us a potential way to break the $\log m\log n$ barrier for problems with $m'>m$ if we can combine it with the previous techniques. 

P\v{a}tra\c{s}cu has used the communication lower bound for lopsided set disjointness, set disjointness problem of a special form, to prove a collection of (static) data structure lower bounds~\cite{Pat08, Pat11}. In this paper, we applied communication \emph{protocol (upper bound)} for sparse set disjointness, set disjointness of a different special form, to prove a dynamic data structure \emph{lower bound}. In some sense, this can be viewed as an analogue of the recent development in duality between algorithm and complexity~\cite{Williams13, Williams14, AWY15} in the communication complexity and data structure world. It would be interesting to see examples where both communication lower bound and upper bound for the exact same problem can be used to prove data structure lower bounds.

On Klee's measure problem, our result is an unconditional lower bound for one certain type of algorithms. From Theorem~\ref{thmdiu}, we can generate a hard input distribution for the sweep-line algorithm, such that if the algorithm only sorts the rectangles, goes through the entire area row by row and computes the number of grids in the union only based on the rectangles intersecting the current row or previous rows, then it cannot beat Bentley's algorithm. However, our hard distribution is not very robust, in the sense that if we do the sweep-line from a different direct, the distribution over inputs becomes really easy. At least, it still shows us what an $o(N\log N)$ time algorithm for computing 2D Klee's measure problem on $[0,N^{2/3}]\times [0,N^{2/3}]$ should \emph{not} look like, if exists. 

\paragraph*{Acknowledgement.}
The author would like to thank Yuqing Ai and Jian Li for introducing Klee's measure problem to me during a discussion, and would like to thank Timothy Chan for telling me the state-of-the-art.

 The author also wishes to thank Ryan Williams for helpful discussions on applications to dynamic graph problems and in paper-writing.

\bibliographystyle{plain}
\bibliography{dynintun}

\appendix

\section{Reduction from Partial Sum}\label{sectredps}
The partial sum problem is to maintain a sequence of $n$ numbers $(A_i)$ over $[n]$, supporting the following operations:
\begin{itemize}
	\item
		\verb+update(i, v)+: set $A_i$ to value $v$;
	\item
		\verb+query(l)+: return $\sum_{i\leq l}A_i$.
\end{itemize}

P\v{a}tra\c{s}cu and Demaine~\cite{PD04a} showed that at least one of the operations needs to take $\Omega(\log n)$ time in the cell-probe model with word size $w=\Theta(\log n)$. 

\begin{proposition}\label{redps}
Any data structure for the \diu{} problem with insertion time $t_i$, deletion time $t_d$ and query time $t_q$ must have $\max\{t_i,t_d,t_q\}\geq \Omega(\log n)$ in the cell-probe model with word size $w=\Theta(\log n)$. 
\end{proposition}

\begin{proof}
Consider the partial sum problem with $\sqrt{n}$ numbers over $[\sqrt{n}]$. Any data structure will take $\Omega(\log n)$ time to update or query. Fix a \diu{} data structure over $[0, n]$, we will use it to solve the partial sum problem. First partition $[0, n]$ into $\sqrt{n}$ blocks of length $\sqrt{n}$ each. The $i$-th block will correspond to number $A_i$. We maintain the invariant that for each number $A_i$, there is an interval in the corresponding block of length equal to the value of $A_i$. 

More specifically, every time we need to update $A_i$ from value $v'$ to $v$, we delete the previous interval $[(i-1)\sqrt{n},(i-1)\sqrt{n}+v']$, then insert a new interval $[(i-1)\sqrt{n},(i-1)\sqrt{n}+v]$. When we need to query the sum of first $l$ numbers, we first insert an interval $[l\sqrt{n},n]$, which covers all blocks from the $(l+1)$-th to the last, then query the length of the union, delete the interval inserted earlier. We know that the temporarily inserted interval has length $(\sqrt{n}-l)\sqrt{n}$. Subtracting it from the answer returned, we get the total length of intervals in the first $l$ blocks, which is exactly the sum of first $l$ numbers.

Every update of partial sum can be implemented using an insertion and a deletion of \diu{}, every query can be implemented using an insertion, a deletion and a query. Therefore, at least one of the operations has to take $\Omega(\log n)$ time. 
\end{proof}

\section{Catalogue of Dynamic Graph Problems in Our Application}\label{sectcata}

The \emph{dynamic \#SCC} problem is to maintain a directed graph $G$, supporting:
\begin{itemize}
	\item
		\verb+insert(u, v)+: insert an edge $(u, v)$;
	\item
		\verb+delete(u, v)+: delete an (existing) edge $(u, v)$;
	\item
		\verb+query()+: return the number of strongly connected components in $G$.
\end{itemize}

The \emph{dynamic planar s-t min-cost flow} problem is to maintain an undirected planar flow network $G$ with edge cost, supporting:
\begin{itemize}
	\item
		\verb+update(u, v, cap)+: update the capacity of (an existing) edge $(u, v)$ to \verb+cap+;
	\item
		\verb+query(f)+: return the min-cost flow from a fixed source $s$ to a fixed sink $t$ with flow value $f$. 
\end{itemize}

The \emph{dynamic weighted s-t shortest path} problem is to maintain a weighted directed graph $G$, supporting:
\begin{itemize}
	\item
		\verb+insert(u, v, w)+: insert an edge $(u, v)$ with weight $w$;
	\item
		\verb+delete(u, v)+: delete an (existing) edge $(u, v)$;
	\item
		\verb+query()+: return the shortest path from a fixed source $s$ to a fixed target $t$.
\end{itemize}

\iftoggle{conf}{
\section{Preliminaries on Cell-probe Model}\label{sectcp_app}
\cellprobepre

\section{Preliminaries on Communication Complexity}\label{sectpre_app}
\sectprecomm

\section{Full Version}
}{}

\end{document}